\documentclass[preprint,10pt]{elsarticle}
\usepackage[a4paper,scale=0.8]{geometry}
\usepackage{amsmath}
\usepackage{amsthm}
\usepackage{amssymb}
\usepackage{bm}  
\usepackage{url}
\usepackage{enumerate}
\usepackage{graphicx}
\usepackage{blkarray}
\usepackage{color}
\usepackage{caption}
\usepackage{subcaption}
\usepackage{fullpage}
\usepackage{sidecap}
\usepackage{wrapfig}
\usepackage{flushend}
\newcommand{\Oh}{\mathcal{O}}

\newcommand{\xc}{\mathop{\mathrm{xc}}} 

\newcommand{\TSP}{{\mathrm{TSP}}}

\newcommand{\PM}{{\mathrm{PM}}}
\newcommand{\HH}{{\mathcal{H}}}

\theoremstyle{plain}
\newtheorem{theorem}{Theorem}
\newtheorem{corollary}{Corollary}
\newtheorem{lemma}{Lemma}
\newtheorem{proposition}{Proposition}

\theoremstyle{definition}

\theoremstyle{remark}

\journal{arXiv}

\begin{document}

\begin{frontmatter}

\title{On the H-Free Extension Complexity of the TSP}

\author[ku]{David Avis} 
\ead{avis@cs.mcgill.ca}

\author[cu]{Hans Raj Tiwary}
\ead{hansraj@kam.mff.cuni.cz}



\address[ku]{Graduate School of Informatics,
   Kyoto University, Sakyo-ku, Yoshida, Kyoto 606-8501, Japan}

\address[cu]{KAM/ITI,
 Charles University,
 Malostransk\'e n\'am. 25,
 118 00 Prague 1, Czech Republic}

\begin{abstract} 
It is known that the extension complexity of the TSP polytope for the 
complete graph $K_n$ is exponential in $n$ even if the subtour inequalities 
are excluded. In this article we study the polytopes formed by removing other 
subsets $\HH$ of facet-defining inequalities of the TSP polytope. In 
particular, we consider the case when $\HH$ is either the set of blossom 
inequalities or the simple comb inequalities. These inequalities are 
routinely used in cutting plane algorithms for the TSP. We show that the 
extension complexity remains exponential even if we exclude these 
inequalities. In addition we show that the extension complexity of polytope 
formed by all comb inequalities is exponential. For our proofs, we introduce a 
subclass of comb inequalities, called $(h,t)$-uniform inequalities, which 
may be of independent interest. 
\end{abstract}


\end{frontmatter}


\section{Introduction} 
A polytope $Q$ is called an extended formulation or an extension of polytope $
 P$ if $P$ can be obtained as a projection of $Q$. Extended formulations are 
of natural interest in combinatorial optimization because even if $P$ has a 
large number of facets and vertices, there may exist a small extended 
formulation for it, allowing one to optimize a linear function over $P$ 
indirectly by optimizing instead over $Q$. Indeed, many polytopes of 
interest admit small extended formulations. (See
\cite{ConfortiCornuejolsZambelli10}, 
for example, for a survey).

Recent years have seen many strong lower bounds on the size of extended 
formulations. In particular, Fiorini et al. \cite{FMPTW} showed 
superpolynomial lower bounds for polytopes related to the MAX-CUT, TSP, and 
Independent Set problems. This was extended to more examples of polytopes 
related to other NP-hard problems having superpolynomial lower bounds
\cite{AT13,PV13}. 
Even though these results are remarkable, they are hardly surprising since 
existence of a small extension for any of these polytopes would have extremly 
unexpected consequences in complexity theory.

Subsequently, Rothvo{\ss} showed that the perfect matching polytope of 
Edmonds does not admit a polynomial sized extended formulation
\cite{Rothvoss14}, 
even though one can separate over it in polynomial time despite the 
polytope having exponentially many vertices and facets. To reconcile this 
apparent lack of power of compact extended formulations to capture even ``easy" 
problems like perfect matching, the authors of this article introduced the 
notion of $\HH$-free extended formulations \cite{AT15}. 

Intuitively, in this setting, given a polytope $P$ (presumably with a high 
extension complexity) and a set of valid inequalities $\HH$, one would like to 
understand the extent to which the inequalities in $\HH$ cause a bottleneck in 
finding a good extended formulation for $P$. More formally, the $\HH$-free 
extension complexity of a polytope $P$ measures the extension complexity of 
the polytope formed by removing the inequalities in $\HH$ from the
facet-defining 
inequalities of $P$. Particularly interesting classes of 
inequalities, for any polytope, are those for which one can construct an 
efficient separation oracle.

Clearly, in this setting, nothing interesting happens if the inequalities to 
be ``removed'' are redundant. In this article, we consider the Traveling 
Salesman Polytope and study its $\HH$-free extension complexity when $\HH$ is
the 
set of simple comb inequalities or the set of $2$-matching inequalities. Both
sets of 
inequalities form important classes  of inequalities for the TSP polytope. 
Whereas efficient separation algorithms are known for the $2$-matching 
inequalities, no such algorithm is known for comb inequalities, which generalize the set of $2$-matching inequalities \cite{PaRa82, FLL06}. 

In this article we identify a parameterized subset of comb inequalities which 
we call \emph{$(h,t)$-uniform comb inequalities} where the parameters require 
a uniform intersection between the handle and all the teeth of the comb. We use 
these inequalities to show that the intersection of comb inequalities defines 
a polytope with exponential extension complexity. Furthermore we show that if 
$\HH$ is a set of valid inequalities for the TSP polytope such that $\HH$ does 
not contain the $(h,t)$-uniform comb inequalities for some values of parameters
$h
$ and $t$, then the $\HH$-free extension complexity of the TSP polytope on $K_n$
 is at least $2^{\Omega(n/t)}$. As corollaries we obtain exponential lower
bounds for 
the $\HH$-free extension complexity of the TSP polytope with respect to
$2$-matching 
inequalities and simple comb inequalities.

The rest of this article is organized as follows. In the next section we 
describe the comb and $2$-matching inequalities and introduce the
$(h,t)$-uniform 
comb inequalities. We also introduce the central tool that we use: subdivided 
prisms of graphs. After a brief motivation for the study of subdivided prisms in
Section \ref{sec:motivation},  we prove our main lemma in Section
\ref{sec:main_lemma}.
We show that over suitably subdivided prisms of the complete graph, there 
exists a canonical way to translate perfect matchings into TSP tours that can 
be done without regard to any specific comb inequality. This 
translation, together with known tools developed in \cite{FFGT12} connecting 
extension complexity with randomized communication protocols gives the 
desired results for the problems of interest. Finally, we discuss applications of the main result in Section
\ref{sec:results}.

\section{Definitions}
Let $P$ be a polytope in $\mathbb{R}^d$. The extension complexity of $P$ -- 
denoted by $\xc(P)$ -- is defined to be the smallest number $r$ such that 
there exists an extended formulation $Q$ of $P$ with $r$ facets.

Let $G=(V,E)$ be a graph. For any subset $S$ of vertices, we denote the edges 
crossing the boundary of $S$ by $\delta(S)$. That is, $\delta(S)$ denotes the
set 
of edges $(u,v)\in E$ such that $|S\cap\{u,v\}|=1.$

The TSP polytope for the complete graph $K_n$ is defined as the convex hull 
of the characteristic vectors of all TSP tours in $K_n$, and is denoted by $
\TSP_n$. Similary, $\PM_n$ denotes the convex hull of all perfect matchings in 
$K_n$. We say that any inequality $a^\intercal x\leqslant b$ is valid for a 
polytope $P$ if every point in $P$ satisfies this inequality. For a point $v$ 
in $P$, the slack of $v$ with respect to a valid inequality $a^\intercal x
\leqslant b$ is defined to be the nonnegative number $b-a^\intercal v.$
 
\subsection{Comb inequalities for TSP}
For a graph $G=(V,E)$, a comb is defined by a subset of vertices $H$ called 
the handle and a set of subsets of vertices $T_i, 1\leqslant i \leqslant k$ 
where $k$ is an odd number at least three. The sets $T_i$ are called the 
teeth. The handle and the teeth satisfy the following properties:

\begin{eqnarray}\displaystyle
 H\cap T_i \neq\emptyset, \\
 T_i\cap T_j=\emptyset,&~~~~\forall i\neq j\\
 H\setminus \bigcup_{i=1}^k T_i\neq\emptyset
\end{eqnarray}

The following inequality is valid for the TSP polytope of $G$ and is called 
the comb inequality for the comb defined by handle $H$ and teeth $T_i$ as above.

$$\displaystyle x(\delta(H))+\sum_{i=1}^{k}x(\delta(T_i))\geqslant 3k+1$$

Gr\"otschel and Padberg \cite{GrotschelP79a} showed that every comb 
inequality defines a facet of $\TSP_n$ for each $n\geqslant 6.$ It is not 
known whether separating over comb inequalities is NP-hard, neither is a 
polynomial time algorithm known.

For a given comb $C$ and a TSP tour $T$ of $G$, the slack between the
corresponding comb 
inequality and $T$ is denoted by $\text{sl}_{\text{comb}}(C,T).$

\subsection*{$2$-matching inequalities} 

A comb inequality corresponding to a handle $H$ and $k$ teeth $T_i$ is called 
a $2$-matching inequality if each tooth $T_i$ has size exactly two. In 
particular this means that $|H\cap T_i|=1 $ and $|T_i \setminus H|=1$ for 
each $1\leqslant   i\leqslant k.$ These inequalities are sometimes also 
referred to as blossom inequalities. Padberg and Rao \cite{PaRa82} gave a 
polynomial time algorithm to separate over the $2$-matching inequalities.

\subsubsection*{Simple comb inequalities}
A comb inequality corresponding to a handle $H$ and $k$ teeth $T_i$ is called 
a simple comb inequality if $|H\cap T_i|=1$ or $|T_i\setminus H|=1$  for each $
1\leqslant i\leqslant k.$ 
Simple comb inequalities contain all the $2$-matching inequalities. It is not 
known whether one can separate over them in polynomial time.

\subsubsection*{$(h,t)$-uniform comb inequalities}
Let us define a subclass of comb inequalities called \emph{$(h,t)$-uniform 
comb inequalities} associated with what we will call \emph{$(h,t)$-uniform 
combs} for arbitrary $1\leqslant h < t.$ A comb, with handle $H$ and $k$ 
teeth $T_i$, is said be $(h,t)$-uniform if $|T_i|=t$ and $H\cap T_i=h,$ for 
all $1\leqslant i\leqslant k.$

\subsection{Odd set inequalities for perfect matching}
Let $V$ denote the vertex set of $K_n$. A subset $U\subset V$ is called an odd
set if the cardinality of $U$ is odd. For every odd set $U$ the following
inequality is valid for the perfect matching polytope $\PM_n$ and is called an
odd set inequality.
$$x(\delta(U)) \geqslant 1$$
For a given odd set $S$ and a perfect matching $M$ of $K_n$, the slack between
the corresponding 
odd set inequality and $M$ is denoted by $\text{sl}_{\text{odd}}(S,M).$

\subsection{$t$-subdivided prisms of a graph}
A prism over a graph $G$ is obtained by taking two copies of $G$ and connecting
corresponding 
vertices. It is helpful to visualise this as stacking the two copies one over
the other and then 
connecting corresponding vertices in the two copies by a vertical edge. A
$t$-subdivided prism is then 
obtained by subdividing the vertical edges by putting $t-2$ extra vertices on
them. See Figure \ref{fig:tower} for an example.
\begin{figure}[!h]
  \begin{center}
    \includegraphics[width=0.2\textwidth]{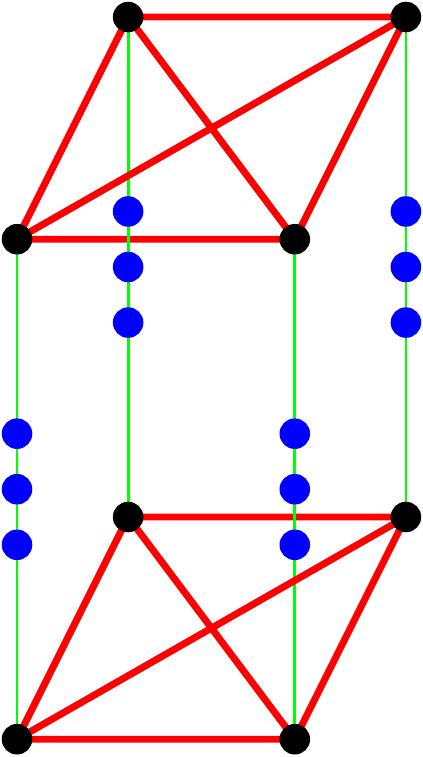}
  \end{center}
	\caption{A $5$-subdivided prism over $K_4.$}
	\label{fig:tower}
\end{figure}

Let $G$ be the $t$-subdivided prism of $K_n$.  Let the vertices of the two
copies be labeled 
$u^1_1,\ldots,u^1_{n}$ and $u^t_1,\ldots,u^t_{n}$. As a shorthand we will denote
the path $u^1_i,u^2_i,\ldots,u^t_i$ as $u^1_i{\leadsto}u^t_i.$ Similarily,
${u^t_i{\leadsto}u^1_i}$ will denote $u_i^t,\ldots,u_i^2,u_i^1.$

The graph $G$ has path ${u^1_i{\leadsto}u^t_i}$ for all $i\in[n]$ and $(u^1_
i,u^1_j), (u^t_i,u^t_j)$ for all $i\neq j, i,j\in [n].$ Thus $G$ has $tn$ 
vertices and $2{n\choose2}+(t-1)n$ edges.

\section{Motivation} \label{sec:motivation}
The motivation for looking at $t$-subdivided prisms stems from a simple 
observation which we state in the form of a proof of the following proposition:

\begin{proposition}
Let $2MP(n)$ be the convex hull of the incidence vectors of all $2$-matchings 
of the complete graph $K_n.$ Then, $\xc(2MP(n))\geqslant 2^{\Omega(n)}.$
\end{proposition}
\begin{proof}
Let $G$ be a graph with $n$ vertices and $m$ edges and let $G'$ be the
$3$-subdivided 
prism of $G $. $G'$ has $3n$ vertices and $2m+2n$ edges. Any 
$2$-matching in $G'$ contains all the vertical edges and thus when restricted 
to a single copy -- say the bottom one --  of $G$ gives a matching in $G$. 
Conversely, any matching in $G$ can be extended to a (not necessarily unique) 
$2$-matching in $G'$.
 
Taking $G$ as $K_n$ we obtain a $G'$ that is a subgraph of $K_{3n}.$ The
$2$-matching 
polytope of $G'$ lies on a face of the $2$-matching polytope of the 
complete graph on $3n$ vertices (corresponding to all missing edges having 
value $0$). Therefore, the extension complexity of the $2$-matching polytope
$2MP(n)$ 
is at least as large as that of the perfect matching polytope. That 
is, $\xc(2MP(n))\geqslant  2^{\Omega(  n)}.$ 
\end{proof}

\noindent The above generalizes to $p$-matching polytopes for arbitrary $p$ in
the 
obvious way, and is probably part of folklore\footnote{W. Cook (private 
communication) attributes the same argument to T. Rothvo{\ss}}.

The generalization of the $3$-subdivided prism to larger subdivisions allows 
us to be able to argue not only about the $2$-matching inequalities -- which are 
the facet-defining inequalities for the $2$-matching polytope -- but also 
about comb inequalities by using the vertical paths as teeth for constructing 
combs.

\section{Main Tools}{\label{sec:main_lemma}}
\subsection{EF-protocols}
Given a matrix $M$, a randomized communication protocol computing $M$ in 
expectation is a protocol between two players Alice and Bob. The players, 
having full knowledge of the matrix $M$, agree upon some strategy. Next, Alice 
receives a row index $i$ and Bob receives a column index $j$. Based on their 
agreed-upon strategy and their respective indices, they exchange a few bits 
and either one of them outputs a non-negative number, say $X_{ij}.$ For 
brevity, we will call such protocols EF-protocols. An EF-protocol is said to 
correctly compute $M$ if for every pair $i,j$ of indices,
$\mathbb{E}[X_{ij}]=M_{ij},$ where $\mathbb{E}[X_{ij}]$ is the expected value of
the random variable $X_{ij}$.

The complexity of the protocol is measured by the number of bits exchanged by 
Alice and Bob in the worst case. It is known that the base-$2$ logarithm of 
the extension complexity of any polytope $P$ is equal to the complexity of 
the best EF-protocol that correctly computes the slack matrix of $P$ \cite{ 
FFGT12}. We will use this fact to show our lower bounds by showing that a 
sublinear EF-protocol for problems of our interest would yield a sublinear
EF-protocol 
for the slack matrix of the perfect matching polytope. First we 
restate some known results about EF-protocols and extension complexity of 
perfect matching polytope in a language that will be readily usable to us.

\begin{proposition}\label{prop:ffgt}{\textbf{\cite{FFGT12}}} 
Let $P$ be a polytope and $ S(P)$ its slack matrix. There exists an EF 
protocol of complexity $\Theta(k)$ that correctly computes $S(P)$ if and only 
if there exists an extended formulation of $P$ of size $2^{\Theta(k)}$.
\end{proposition}

Combining lower bounds by Rothvo{\ss} \cite{Rothvoss14} with the above 
mentioned equivalence by Faenza et al. \cite{FFGT12}, it is easy to see that 
no sublinear protocol computes the slack matrix of the perfect matching 
polytope.

\begin{proposition}\label{prop:rothvoss-pm}{\textbf{\cite{Rothvoss14}}} 
Any EF-protocol that correctly computes the slack matrix of the perfect 
matching polytope of $K_n$ requires an exchange of $\Omega(n)$ bits. 
\end{proposition}

\subsection{Uniform combs of odd sets}
Let $n$ and $t$ be positive integers. In the rest of the article we will 
assume that $n$ is a multiple of $t$. Since we are interested in asymptotic 
statements only, this does not result in any loss of generality. Let $G$ be 
the $t$-subdivided prism of $K_{n/t}$ for some $t\geqslant 2.$ Given an odd 
set $S$ and a perfect matching $M$ in $K_{n/t}$, and arbitrary $1\leqslant h <
t$, we are interested in constructing a comb $C$ and a TSP 
tour $T$ in $K_n$ such that the following conditions hold:

\begin{enumerate}
 \item[\textbf{(C1):}] $C$ is a $(h,t)$-uniform comb.
 \item[\textbf{(C2):}] $C$ depends only on $S$ and $2$ edges of $M.$
 \item[\textbf{(C3):}] $T$ depends only on $M.$
 \item[\textbf{(C4):}]
$\text{sl}_{\text{comb}}(C,T)=\text{sl}_{\text{odd}}(S,M).$
\end{enumerate}


If such a pair $(C,T)$ of a comb and a TSP tour is shown to exist for every 
pair $(S,M)$ of an odd set and a perfect matching, then we can show that any 
EF-protocol for computing the slack $\text{sl}_{\text{comb}}(C,T)$ can be 
used to construct an EF-protocol for computing $\text{sl}_{\text{odd}}(S,M)$ 
due to condition (C4). Furthermore, due to conditions (C2) and (C3) the 
number of bits required for the later protocol will not be much larger than 
the number of bits required for the former, as $C$ can be locally constructed 
from $S$ after an exchange of two edges, and $T$ can be locally constructed 
from $M.$

Now we show that such a pair does exist if at least two edges of $M$ are 
contained in $S$ and $|S|\geqslant 5$. 

\begin{lemma}\label{lem:parsimonious_ct} 
Let $(S,M)$ be a pair of an odd set and a perfect matching in $K_{n/t}$, and 
let  $1\leqslant h<t$. Suppose that $|S| \geqslant 5$, and let $w_1,w_2, w_3, 
w_4\in S$ be distinct with $(w_1,w_2)$ and $(w_3,w_4)$ in $M$. Then, there 
exists a pair $(C,T)$ of a comb $C$ and a TSP tour $T$ in $K_{n}$ 
satisfying the four conditions (C1)--(C4).
\end{lemma}
\begin{proof}

Let $|S|=s$. For simplicity of exposition, we assume that the vertices of $S$ 
are labeled $w_1,\ldots,w_s.$ By $w_i^j$, we denote the copy of $w_i$ in the $
j$-th layer of the $t$-subdivided prism over $K_{n/t}.$

The comb $C$ is constructed as follows. The handle $H$ is obtained by taking 
all vertices in $ S$ and the copies $w^2_1,\ldots,w^t_1$ and $w^2_3,\ldots,w^t
_3.$ For every other vertex $w\in S$ the vertices $w^2, \ldots,w^h$ are also 
added to $H.$  The teeth $T_i$ are formed by pairing each vertex $v$ in $
S\setminus\{w_1,w _3 \}$ with its copies $v^2,\ldots,v^t$ producing $s-2$ 
teeth. See Figure \ref{fig:combs_from_oddset} for an illustration. Since
$s\geqslant 5$ is odd, the number of teeth is 
odd and at least $3$. Thus, the constructed comb is $(h,t)$-uniform satisfying
conditions (C1) and (C2), and 
the corresponding comb inequality is 
\begin{eqnarray}\label{eqn:comb}\displaystyle x(\delta(H))+\sum_{i=1}^
{s-2}x(\delta(T_i))\geqslant 3(s-2)+1.\end{eqnarray}

\begin{SCfigure}
  \centering
  	\caption{Construction of a comb from given odd set: The odd set
consists 
of $5$ vertices displayed as big filled circles in the bottom copy. The 
corresponding handle consists of all vertices represented by filled circles. 
The teeth are represented by the vertical ellipsoidal enclosures.\newline The 
big circles represent vertices of the original graph and their top copies. 
The small circles represent the $h$-th copy, while the other copies have been 
omitted here. Bold edges at the bottom are matching edges. All other edges 
displayed are just for illustration of the relationship of various copies of 
vertices. }
    \includegraphics[width=0.4\textwidth]{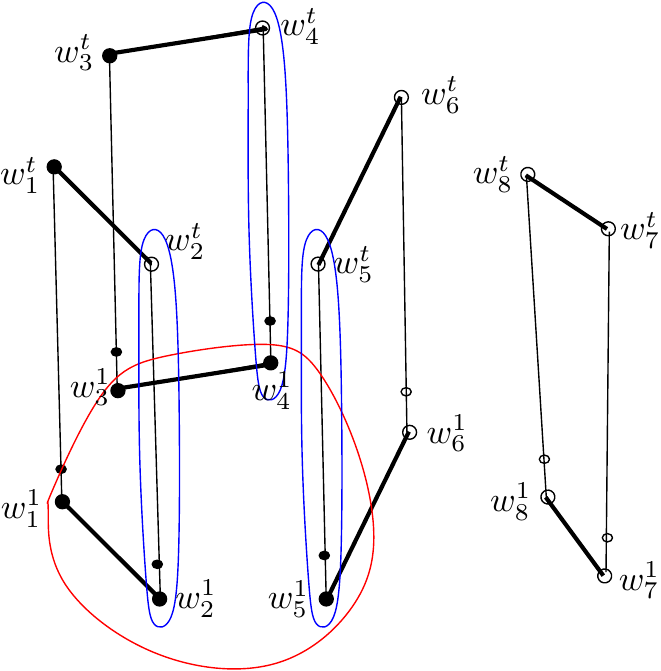}
	\label{fig:combs_from_oddset}
\end{SCfigure}

To construct a tour $T$ from the given perfect matching $M$ such that conditions
(C3) and (C4) are satisfied, we start with a subtour
$({w^1_1{\leadsto}w^t_1},{w^t_3{\leadsto}w ^1_3}
,w^1_4{\leadsto}w^t_4,{w^t_2{\leadsto}w^1_2} ,w^1_1).$ 
At each stage we maintain a 
subtour that contains all matching edges on the induced vertices in the lower 
copy, the edge $(w^t_1,w^t _3)$, and at least one top edge different from
$(w^t_1,w^t_3).$ Clearly the starting subtour satisfies these requirements. As
long 
as we have some matching edges in $M$ that are not in our subtour, we pick an 
arbitrary edge $(w_a,w_b)$ in $M $ and extend our subtour as follows. Select 
a top edge $ (w^ t_q,w^t_r)$ different from $(w^t_1,w^t_3),$ remove the edge 
and add the path $
(w^t_q,{w^t_a{\leadsto}w^1_a},{w^1_b{\leadsto}w^t_b},w^t_r).$ 
The new subtour contains the selected perfect matching edge $(w^1_a,w^1 _b)
,$ the paths ${w^1_a{\leadsto}w^t_a}$ and ${w^1_b{\leadsto}w^t_b}$ and has 
one more top edge distinct from $(w^t_1,w^t_3)$ than in the previous subtour. 
See Figure \ref{fig:tour_from_matching} for an example.

\begin{figure}[t]
        \centering
        \begin{subfigure}[b]{0.32\textwidth}
                \includegraphics[width=\textwidth]{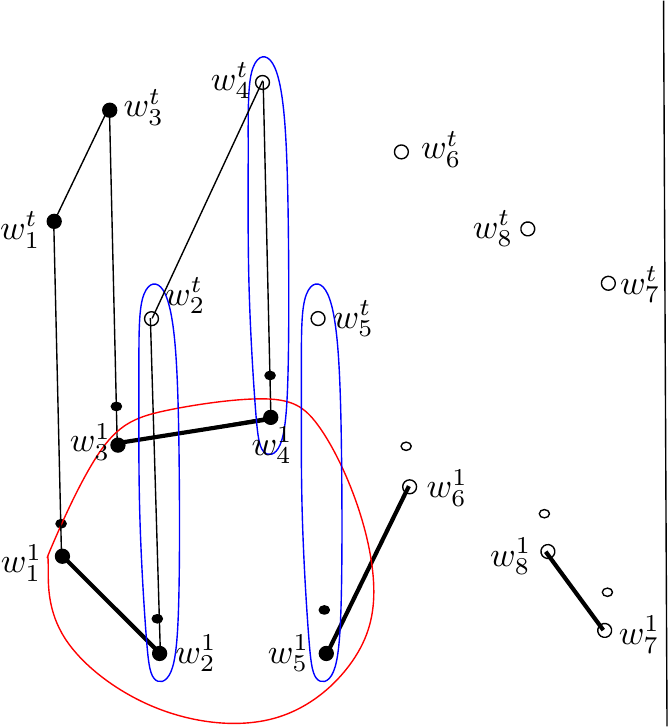}
                \caption{The initial tour going through
$w_1^1,w_1^t,w_3^t,w_3^1,w_4^1,w_4^t,w_2^t,w_2^1,w_1^1$}
                \label{fig:tour_from_matching_step1}
        \end{subfigure}
        \begin{subfigure}[b]{0.32\textwidth}
                \includegraphics[width=\textwidth]{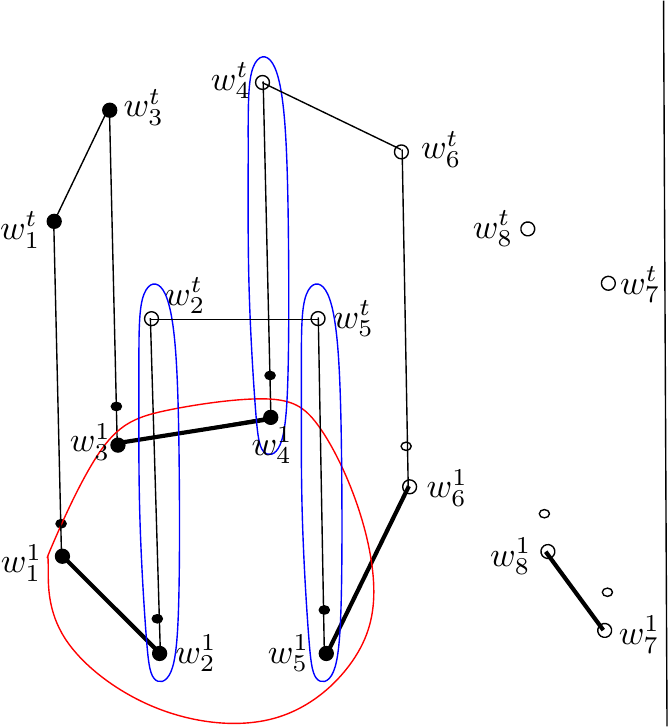}
                \caption{Adding a new matching edge $(w_5,w_6)$ by removing
$(w_2^t,w_4^t)$}
                \label{fig:tour_from_matching_step2}
        \end{subfigure}
        \begin{subfigure}[b]{0.32\textwidth}
                \includegraphics[width=\textwidth]{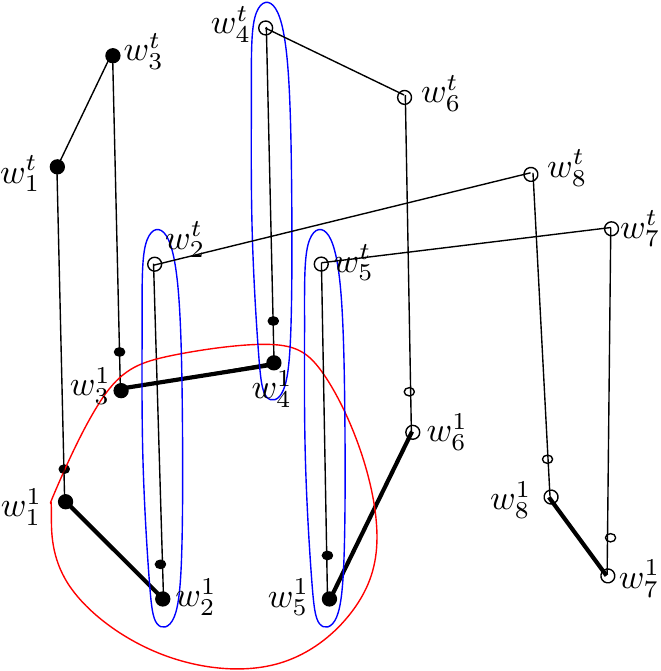}
                \caption{The final tour}
                \label{fig:tour_from_matching_step3}
        \end{subfigure}
        \caption{Constructing a TSP tour from a perfect
matching.}\label{fig:tour_from_matching}
\end{figure}

At the completion of the procedure, we have a TSP tour that satisfies the
following properties:
\begin{enumerate}
 \item Each edge of $M$ is used in the tour.
 \item Each vertical path ${w^1_i{\leadsto}w^t_i}$ for all $i\in[n]$ is used in
the tour.
 \item Edge $(w^t_1,w^t_3)$ is used in the tour.
\end{enumerate}

From the construction, edges in $|\delta(H)\cap T|$ are
precisely the edges in $|\delta(S)\cap M|$ together with $s{-}2$ other edges
exiting the comb: one through each
of the $s{-}2$ teeth. Therefore, $|\delta(H)\cap T|=|\delta(S)\cap M|+s-2$. 
Also, the tour $T$ enters and exits each teeth precisely once so
$|\delta(T_i)\cap T|=2$ 
for each of the $s{-}2$ teeth. Substituting these values in the inequality
\ref{eqn:comb}, we obtain the slack $\text{sl}_{\text{comb}}(C,T)=
|\delta(S)\cap M|+ (s-2)+2(s-2)-3(s-2)-1 =\text{sl}_{\text{odd}}(S,M).$ 
This completes the proof because the pair $(C,T)$  satisfies conditions
(C1)--(C4).

\end{proof}

We are finally ready to state the main Lemma of this article. Using the
existence of the pair $(C,T)$ as described earlier and the fact that any
EF-protocol for the perfect matching polytope requires an exchange of a linear
number of bits, we will lower bound the number of bit exchanged by any
EF-protocol computing the slack of $(h,t)$-uniform comb inequalities with
respect to TSP tours. In the next Section we will use this Lemma multiple times
by fixing different values for the parameters $h$ and $t.$

\begin{lemma}\label{lem:ef-protocol-kl-uniform-to-pm}
Any EF-protocol computing the slack of $(h,t)$-uniform comb inequalities with
respect to the TSP tours of $K_{n},$ requires an exchange of $\Omega(n/t)$ bits.
Equivalently, the extension complexity of the 
polytope of $(h,t)$-uniform comb inequalities is $2^{\Omega(n/t)}.$ 
\end{lemma}
\begin{proof}
Due to Proposition \ref{prop:rothvoss-pm}, it suffices to show if such a 
protocol uses $r$ bits, then an EF-protocol for the perfect matching 
polytope for $K_{n/t}$ can be constructed, that uses $r+\Oh(\log{(n/t)})$ bits.  The protocol for 
computing the slack of an odd set inequality with respect to a perfect 
matching in $K_{n/t}$ works as follows.

Suppose Alice has an odd set $S$ in $K_{n/t},$ with $|S|=s,$ and Bob has a 
matching $M$ in $K_{n/t}.$ The slack of the odd-set inequality corresponding 
to $S$ with respect to matching $M$ in the perfect matching polytope for $K_{n
/t}$ is $| \delta(S)\cap M|-1.$ 

We assume that $s\geqslant 5.$ Otherwise, Alice can send the identity of the 
entire set $S$ with at most $4\log{(n/t)}$ bits and Bob can output the slack 
exactly.

Alice first sends an arbitrary vertex $w_1\in S$, to Bob. Bob replies with the 
matching vertex of $w_1,$ say $w_2.$ Alice then sends another arbitrary 
vertex $w_3\in S, w_3\neq w_2$ to Bob who again replies with the matching 
vertex for $w_3$, say $w_4.$ So far the number of bits exchanged is $4
\left\lceil\log{(n/t)}\right\rceil$. 

Now there are two possibilities: either at least one of the vertices $w_2,w_4$ 
is not in $S$, or both $w_2 ,w_4$ are in $S.$  Alice sends one bit to 
communicate which of the possibilities has occurred and accordingly they 
switch to one of the two protocols as described next.

In the former case, Alice has identified an edge, say $e$, in $\delta(S)\cap 
M. $ Now Bob selects an edge $e'$ of his matching uniformly at random (i.e. 
with probability $2/n$) and sends it to Alice. If $e'$ is in $\delta(S) 
\setminus\{e\}$, Alice outputs $n/2.$ Otherwise, Alice outputs zero. The
expected contribution by edges in $(\delta(S)\cap M)\setminus\{e\}$ is then
exactly one while the expected contribution of all other edges is zero.
Therefore the expected output is $|\delta(S)\cap M|-1,$ and the number 
of bits exchanged for this step is $\left\lceil\log{m}\right\rceil$ where $m$ 
is the number of edges in $K_{n/t}.$ Thus the total cost in this case is $\Oh(
\log{(n/t)})$ bits.

In the latter case, the matching edges $(w_1,w_2)$ and $(w_3,w_4)$ lie inside 
$S.$ Alice constructs a comb $C$ in the $t$-subdivided prism of $K_{n/
t}$, 
and Bob a TSP tour $T$ in the $t$-subdivided prism of $K_{n/t}$ such that $(
C,T)$
  satisfies conditions (C1)--(C4). By Lemma \ref{lem:parsimonious_ct} they 
can do this without exchanging any more bits.  Since $ \text{sl}_{ \text{  
comb}}( C,T)=\text{sl}_{\text{odd}}(S,M)$,  they proceed to compute the 
corresponding slack with the new inequality and tour, exchanging $r$ bits. 
The total number of bits exchanged in this case is $r+4\left\lceil\log{(n/t)}
\right\rceil+1=r+\Oh(\log{(n/t)}).$ 
\end{proof}

\section{Applications}\label{sec:results}
In this section we consider the extension complexity of the polytope of comb 
inequalities and $\HH$-free extension complexity of the TSP polytope when $\HH$
is the set of simple comb inequalities. As we will see, the results in this
section are obtained by instantiating Lemma
\ref{lem:ef-protocol-kl-uniform-to-pm} with different values of the parameters
$h$ and $t$.

\subsection{Extension complexity of Comb inequalities}
We show that the polytope defined by the Comb inequalities has high extension
complexity. 

\begin{theorem}\label{thm:uniform_comb_xc}
Let COMB$(n)$ be the polytope defined by the intersection of all comb 
inequalities for $\TSP_n.$ Then $\xc(\text{COMB}(n))\geqslant
2^{\Omega(n)}.$
\end{theorem}
\begin{proof}
Suppose there exists an EF-protocol that computes the slack of COMB$(n)$ that 
uses $r$ bits. Since $(1,2)$-uniform comb inequalities are valid for $\TSP_n$
 we can use the given protocol to compute the slack of these inequalities 
with respect to the TSP tours of $K_{n}$ using $r$ bits. Then, using Lemma 
\ref{lem:ef-protocol-kl-uniform-to-pm}, the slack matrix of the perfect 
matching polytope for $K_{n/2}$ can be computed using $r+\Oh(\log n)$ bits. 
By Proposition \ref{prop:rothvoss-pm}, this must be $\Omega(n)$. Finally, by 
Proposition \ref{prop:ffgt} this implies that $\xc(\text{COMB}(n))\geqslant 2^
{\Omega(n)}.$
\end{proof}

\subsection{$\HH$-free extension complexity}
Let $\mathcal{C}_{h,t}$ be the set of $(h,t)$-uniform comb inequalities for 
fixed values of $h$ and $t.$ Observe that, since at least three teeth are
required to define a comb and the handle must contain some vertex not in any
teeth, for $(h,t)$-uniform combs on $n$ vertices we must have $t\leqslant
\left\lfloor\frac{n-1}{3}\right\rfloor.$ So for any 
values of $1\leqslant h < t\leqslant \left\lfloor\frac{n-1}{3} \right\rfloor,$ 
the set $\mathcal{C}_{h,t}$ is a nonempty set of facet-defining inequalities 
for $\TSP_n$, and for any other values of $h$ and $t$ the set $ 
\mathcal{C}_{h,t}$ is empty. 

\begin{theorem}\label{thm:hfree_main} 
If $\HH$ is a set of inequalities valid for the polytope $\TSP_n,$ such that
$\HH 
\cap\mathcal{C}_{h,t}=\emptyset$ for some nonempty $\mathcal{C}_{h,t}$, then 
the $\HH$-free extension complexity of $\TSP_n$ is at least $2^{\Omega(n/t)}.$ 
\end{theorem} 
\begin{proof} 
Let $1\leqslant h<t$ be integers such that $\HH\cap\mathcal{C}_{h,t}=\emptyset.$
That is, the set $\HH$ does not contain any $(h,t)$-uniform comb inequalities.
Let 
$P$ be the polytope formed from $\TSP_n$ by throwing away any facet-defining 
inequalities that are in $\HH$. Then, any EF-protocol computing the 
slack matrix of $P$ correctly must use $\Omega(n/t)$ bits due to Lemma 
\ref{lem:ef-protocol-kl-uniform-to-pm}. The claim then follows from Proposition 
\ref{prop:ffgt}. 
\end{proof}

The above theorem shows that for every set $\HH$ of valid inequalities of
$\TSP_n$, 
if the extension complexity of the TSP polytope becomes polynomial after 
removing the inequalities in $\HH$, then $\HH$ must contain some inequalities 
from every $(h,t)$-uniform comb inequality class, for all $t=o(n/\log{n}).$ The 
theorem can easily be made stronger by replacing the requirement $\HH 
\cap\mathcal{C}_{h,t}=\emptyset$ with $|\HH \cap\mathcal{C}_{h,t}|\leqslant 
\text{poly}(n)$. (See the discussion about $\HH$-free extension complexity of
$\TSP_n$ 
with respect to subtour inequalities in Avis and Tiwary \cite{AT15} for 
clarification.) 

We can use the above theorem to give lower bounds for $\HH$-free extension 
complexity of the TSP polytope with respect to important classes of valid 
inequalities by simply demonstrating some class of $(h,t)$-uniform comb 
inequalities that has been missed.

\subsubsection*{$2$-matching inequalities}
\begin{corollary}\label{cor:2_matching}Let $P$ be the polytope obtained by 
removing the $2$-matching inequalities from the TSP polytope. Then,
$\xc(P)=2^{\Omega(n)}.$
\end{corollary}
\begin{proof}
The $2$-matching inequalities are defined by combs for which each tooth has 
size exactly two. Therefore the set of $(1,3)$-uniform combs are not
$2$-matching 
inequalities, and Theorem \ref{thm:hfree_main} applies.
\end{proof}

\subsubsection*{Simple comb inequalities}
\begin{corollary}\label{cor:simple_comb} 
Let $P$ is the polytope obtained by removing the set of simple comb 
inequalities from the TSP polytope. Then, $\xc(P)=2^{\Omega(n)}.$
\end{corollary}
\begin{proof}
Recall that a comb is called simple if $|H\cap T_i|=1$ or $|T_i\setminus H|=1$
 for all $1\leqslant i\leqslant k$ where $k$ is the (odd) number of teeth in 
the comb and $H$ is the handle. Clearly,  $(2,4)$-uniform combs are not 
simple and Theorem \ref{thm:hfree_main} applies.
\end{proof}

As mentioned before, simple comb inequalities define a superclass of
$2$-matching 
inequalities and a polynomial time separation algorithm is known for
$2$-matching 
inequalities. Althought a similar result was claimed for simple comb 
inequalities, the proof was apparently incorrect, as pointed out by Fleischer 
et al. \cite{FLL06}. This latter paper includes a polynomial time separation 
algorithm for the wider class of simple domino-parity inequalities that we do 
not consider here.

We leave as an open problem whether there exists a polynomial time separation 
algorithm for the $(h,t)$-uniform comb inequalities.

\section*{Acknowledgment}
Research of the first author is supported by a Grant-in-Aid for Scientific 
Research on Innovative Areas -- Exploring the Limits of Computation, MEXT, 
Japan. Research of the second author is partially supported by GA~\v{C}R grant
P202-13/201414.

  \bibliographystyle{elsarticle-num} 
  \bibliography{hfree_tsp}

\end{document}